\documentclass[11pt]{article}
\usepackage{amsmath, amssymb, amsthm, verbatim, fullpage, charter, complexity, thm-restate, color, esvect, graphicx, hyperref, dsfont}

\usepackage[T1]{fontenc}
\newtheorem{thm}{Theorem}
\newtheorem{lemma}{Lemma}
\newtheorem{observation}{Observation}
\newtheorem{corollary}{Corollary}

\def\reals{{\mathbb R}}
\def\eps{{\varepsilon}}
\def\etal{{\it et al.~}}

\title{Implicit representations via the polynomial method}

\author{
Jean Cardinal\thanks{%
Universit\'e libre de Bruxelles (ULB), Brussels, Belgium; {\sf jean.cardinal@ulb.be};
{https://orcid.org/0000-0002-2312-0967}{}}
\and
Micha Sharir\thanks{%
School of Computer Science, Tel Aviv University, Tel~Aviv, Israel; {\sf michas@tau.ac.il};
{https://orcid.org/0000-0002-2541-3763}{}}
}

\begin{document}
\maketitle
\sloppy

\begin{abstract}
Semialgebraic graphs are graphs whose vertices are points in $\reals^d$, and adjacency between two vertices is
determined by the truth value of a semialgebraic predicate of constant complexity.
We show how to harness polynomial partitioning methods to construct compact adjacency labeling schemes 
for families of semialgebraic graphs.

That is, we show that for any family of semialgebraic graphs, given a graph on $n$ vertices in this family, we can
assign a label consisting of $O(n^{1-2/(d+1) + \eps})$ bits to each vertex (where $\eps > 0$ can be
made arbitrarily small and the constant of proportionality depends on $\eps$ and on the complexity
of the adjacency-defining predicate), 
such that adjacency between two vertices can be determined solely from their two labels,
without any additional information. 
We obtain for instance that unit disk graphs and segment intersection graphs have such labelings with labels of
$O(n^{1/3 + \eps})$ bits.
This is in contrast to their natural implicit representation consisting of the coordinates of the 
disk centers or segment endpoints, which sometimes require exponentially many bits.
It also improves on the best known bound of $O(n^{1-1/d}\log n)$ for $d$-dimensional 
semialgebraic families due to 
Alon (\textit{Discrete Comput. Geom.}, 2024), a bound that holds more generally for graphs with shattering functions bounded by a degree-$d$ polynomial. Our labeling scheme is efficient in the sense 
that not only adjacency between two vertices can be decided in time linear in the size of their labels,
but the labels can be computed in subquadratic time on a real RAM from the input points and the 
semialgebraic adjacency predicate, using recent polynomial partitioning algorithms.
  
We also give new bounds on the size of adjacency labels for other families of graphs.
In particular, we consider semilinear graphs, which are semialgebraic graphs in which the 
predicate only involves linear polynomials. We show that semilinear graphs have adjacency labels 
of size $O(\log n)$. We also prove that polygon visibility graphs, which are not semialgebraic 
in the above sense, have adjacency labels of size $O(\log^3 n)$.
\end{abstract}

\section{Introduction}

Finding compact encodings of graphs is a standard topic in computer science and discrete mathematics, see for
instance the seminal monograph of Spinrad~\cite{MR1971502}. In this paper, we explore a special kind of distributed
encoding, that associates a binary label, of relatively small size, to each vertex, in such a way that 
adjacency between two vertices can be deduced solely from their labels, without any additional information 
about the graph. We focus on graphs defined by semialgebraic predicates, which are commonly encountered 
in computational geometry and other domains. We first carefully define our objects of study, then summarize our contributions.
Our results are stated in Section~\ref{sec:results}.

\subsection{Adjacency labeling schemes}

We only consider simple, undirected graphs.
A family of graphs is said to have a \emph{$t(n)$-bit adjacency labeling scheme}, where $n$ is the 
number of vertices of the graph, if there exists a computable boolean function 
$A:\{ 0,1\}^*\times \{0,1\}^*\mapsto \{\mathrm{true}, \mathrm{false}\}$ such that for 
every graph $(V,E)$ on $n$ vertices in this family, there exists a map $\ell:V\mapsto \{0,1\}^{t(n)}$ 
such that $(u,v)\in E$ if and only if $A(\ell(u),\ell(v))$ is true.
We may further require that the function $A$ be computable efficiently, that there exist an 
efficient algorithm to produce the labels given a graph in the family, and that adjacency
between two given vertices can be determined efficiently, but the problem
makes sense even without these additional requirements.
Labeling schemes are also referred to as \emph{implicit representations}, and have been the topic of 
intensive research since the seminal results of Muller~\cite{M88} and Kannan, Naor, and Rudich~\cite{MR1186827}. 
See Figure~\ref{fig:example} for an example.

\begin{figure}
  \begin{center}
    \includegraphics{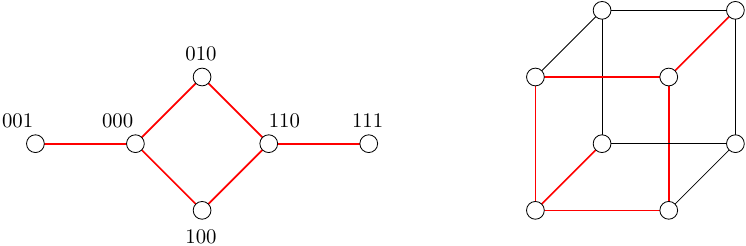}
  \end{center}
  \caption{\label{fig:example}
  {\small{\sf Example of an adjacency labeling of a graph (left).
    Here two vertices $u$ and $v$ are adjacent, hence $A(\ell (u), \ell(v))$ is true, if and only if the 
    Hamming distance between the two labels $\ell(u)$ and $\ell(v)$ is exactly equal to one.
    The graph is therefore an induced subgraph of the cube (shown in red, right).}}}
\end{figure}

Note that if a family has a $t(n)$-bit labeling scheme for which the labeling function $\ell$ is
injective, then for every $n$, there exists a \emph{universal graph} $U_n$ on $2^{t(n)}$ vertices, such that any
$n$-vertex graph in the family is an induced subgraph of $U_n$~\cite{MR1186827}.
Hence upper bounds on the size of adjacency labeling schemes can be thought of as upper bounds 
on the size of universal graphs. 

Clearly, the family of all graphs has an $O(n)$-bit adjacency labeling scheme.
One may, for instance, assign to each vertex the binary word formed by its row in 
the adjacency matrix of the graph, together with its index (column number) in the matrix.
Then, to determine whether $(u,v)$ is an edge, we check whether the $v$-entry of the row of $u$ is $1$.
It is actually known that the family of all graphs has an $n/2+O(1)$ adjacency labeling scheme, which is optimal, up
to an additive constant~\cite{MR3613451, MR3899158}. Results on trees~\cite{MR3702457}, 
cycles and paths~\cite{MR4107731}, and bounded-degree graphs~\cite{MR3685868, MR3238387} are also known.
It is also known that planar graphs have adjacency labeling schemes of size $(1+o(1))\log_2 n$~\cite{MR4402353}.

The \emph{speed} of a family of graphs is the function $s(n)$ that maps an integer $n$ to the 
number of graphs on the vertex set $[n]$ in the family. Observe that if there exists a $t(n)$-bit adjacency labeling scheme for a family of graphs, then $s(n) \leq 2^{n t(n)}$.
Therefore, for a family of speed $s(n)$, the function $\frac 1n \log s(n)$ is a lower bound on the 
minimum size of an adjacency labeling scheme.

A family of graphs is said to be \emph{hereditary} if for any graph in the family, all its induced 
subgraphs are also in the family. Kannan et al.~\cite{MR1186827} conjectured that every hereditary 
family with speed $2^{O(n\log n)}$ has an $O(\log n)$-bit adjacency labeling scheme. Hatami and 
Hatami~\cite{MR4537287} refuted the conjecture by showing that for every $\delta > 0$, there is a 
hereditary family of graphs with speed $2^{O(n\log n)}$ but with no $n^{1/2-\delta}$-bit adjacency labeling scheme.

\subsection{Semialgebraic graphs}

A simple graph $G=(V,E)$ is a \emph{semialgebraic graph} of description complexity $t$ and dimension $d$ if the
vertices in $V$ can be mapped to points in $\reals^d$, so that the presence of an edge is defined by 
a Boolean combination of the sign patterns of $t$ $2d$-variate polynomial functions, of degrees at most $t$,
of the corresponding pair of image points.

More precisely, we have $V\subset\reals^d$, and there exist $t$ $2d$-variate polynomials 
$f_1,\ldots, f_t\in \reals[x_1,\ldots, x_d, y_1,\ldots ,y_d]$,
each of degree at most $t$, and a Boolean function $\Phi$ on $3t$ variables 
such that for any $u,v \in V$,
\[
(u,v) \in E \Leftrightarrow \Phi\left(\{f_i(u,v)<0,f_i(u, v)=0, f_i(u, v)\leq 0\}_{i\in [t]}\right).
\]

Note that if the graph is undirected, as we generally assume in this work, the predicate has to be symmetric, 
in the sense that exchanging $u$ and $v$ yields the same truth value. 
A family of graphs is said to be a \emph{$d$-dimensional semialgebraic family} if every graph in the family is
semialgebraic with the same parameters $d, t, f_1, \ldots ,f_t$, and $\Phi$, where $d$ and $t$ are constants
(so the dependence on $t$, the $f_i$, and $\Phi$ is suppressed in this notation).

Semialgebraic families include many well-studied families, such as geometric intersection graphs.
Among the simplest examples, one can think of \emph{unit disk graphs}, which are intersection graphs 
of unit disks in the plane. There every vertex $v$ can be assigned the coordinates of 
the center of the corresponding disk, and the function $\Phi$ simply encodes the condition that two 
centers are at distance at most two. Many standard problems and tools from graph theory, such as 
regularity lemmas~\cite{MR2983197,MR3585030,MR4863535}, Erd\H{o}s-Hajnal properties~\cite{MR2156215},
Ramsey-Tur\'an numbers~\cite{MR3907792}, {Z}arankiewicz's problem~\cite{MR3646875}, and
Ramsey-type results~\cite{MR3217709} have been studied in a semialgebraic setting.

When the polynomials $f_1,f_2,\ldots ,f_t$ in the above definition are all linear, the graphs are 
said to be \emph{semilinear}. Families of semilinear graphs include many well-studied families of 
graphs, such as interval graphs, permutation graphs, and bounded-boxicity graphs; see 
Section~\ref{sec:semilinear} for details. Again, several standard problems in graph theory have been 
studied recently in the semilinear setting~\cite{MR4308822, MR4591829}.

The question of the existence of implicit representations of small size 
for semialgebraic families has been recently raised by Alon~\cite[Section 3, paragraph 3]{MR4800729}.
See Section~\ref{sec:disc} for a discussion of known results.

\subsection{Adjacency labeling schemes via biclique decompositions}

A \emph{biclique} is a complete bipartite graph.
A \emph{biclique decomposition} of a graph is a partition of its edge set into complete bipartite subgraphs.
Biclique decompositions are ubiquitous in computational geometry, where they are typically obtained as 
a byproduct of offline range searching data structures~\cite{AES24,MR1739610,BK03,CY25,cha06,E96,MR1471987,MR1220545}.
Biclique decompositions are also useful as compressed representations of graphs on which some algorithms 
(such as breadth-first search) can be run efficiently~\cite{CCCK24,FM95}.

The \emph{size of a biclique decomposition} is the sum of the number of vertices of the bicliques in the decomposition.
Note that this is not the same as the number of bicliques of the decomposition, a parameter which has also been studied,
and whose behavior, even asymptotic, may be different from the size.
The minimum size of a biclique decomposition of a given graph is sometimes referred to as the \emph{representation
complexity} of the graph.
Indeed, up to logarithmic factors, this size is an upper bound on the number of bits needed to encode the graph;
see below.
Chung, Erd\H{o}s, and Spencer~\cite{MR820214} proved that every $n$-vertex graph has a biclique decomposition of size at most $((\log e)/ 2 + o(1))n^2/\log n$, which is tight up to a factor $1/e$.
Erd\H{o}s and Pyber~\cite{MR1452952} further showed that every graph has a biclique decomposition such that every vertex belongs to at most $O(n/\log n)$ bicliques of the decomposition.
The size of a biclique decomposition also plays a key role in bounds for Zarankiewicz's problem, on the maximum number of edges
of bipartite graphs forbidding large bicliques~\cite{MR3646875}. Indeed, it can easily be shown that if a $K_{t,t}$-free
bipartite graph has a biclique cover of size $s$, then it has at most $st$ edges.

Our key observation is that adjacency labeling schemes can be defined from biclique decompositions as follows:
For each vertex $v$, give the list of bicliques $v$ belongs to, and for each of these bicliques, the side of the
bipartition that contains $v$. The adjacency test simply consists of intersecting the two given lists, 
and checking that the two vertices are on opposite sides of one of their common bicliques, if any.
An example is given in Figure~\ref{fig:example2}.

The length of these labels in bits is $\nu (n)\cdot O(\log n)$, where $\nu (n)$ is an upper bound on the number of
bicliques a vertex belongs to. The factor $O(\log n)$ is the number of bits needed to identify a biclique.
A single bit per biclique suffices to indicate the side containing the vertex.
It is also immediate that by sorting the list of bicliques in increasing order of their labels, 
adjacency can be tested in time linear in the length of the lists that contain the pair of vertices.
The result of Erd\H{o}s and Pyber~\cite{MR1452952} yields an $O(n)$ bound on the label size for arbitrary graphs.
Note that this labeling scheme needs not be injective.
For instance, if the whole graph is a biclique, then only two distinct labels are needed.
If we insist on having an injective scheme, we can simply append a $\lceil\log n\rceil$-size 
identifier of a vertex to its label, so that all labels are distinct.

A related technique, with a constant number of graphs covering each vertex, has been used for adjacency labeling under the name \emph{locally bounded coverings}~\cite{MR3279268}.

\begin{figure}
  \begin{center}
    \includegraphics[page=2]{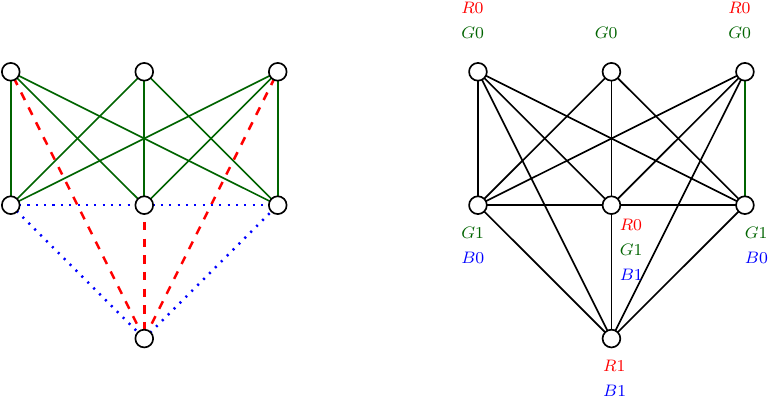}
  \end{center}
  \caption{\label{fig:example2}{\small{\sf Example of an adjacency labeling of a graph with 16 edges (right) obtained
  from a biclique decomposition of size 14 (left). The label $\ell(v)$ of a vertex $v$ consists of the list of
  bicliques $v$ is contained in, from the set $\{R,G,B\}$, together with an additional bit indicating on which side of
  the biclique the vertex $v$ lies.}}}
\end{figure}

\subsection{Our results}
\label{sec:results}

The main result of this paper is that the recent polynomial partitioning strategies used in 
biclique decomposition algorithms for semialgebraic graphs can be adapted to ensure that the 
upper bound $\nu (n)$ on the number of bicliques to which every vertex belongs is strongly sublinear, thereby providing a semialgebraic version of the Erd\H{o}s-Pyber Theorem~\cite{MR1452952}.
As a consequence, we can give a sublinear bound 
on adjacency labels for semialgebraic families. 

\begin{restatable}{thm}{semialgebraic}
  \label{thm:main}
  $d$-dimensional semialgebraic families have an $O(n^{1-2/(d+1) + \eps})$-bit adjacency labeling scheme, where $\eps > 0$ can be made arbitrarily small and the constant of proportionality depends on $\eps$ and on the complexity of the adjacency-defining predicate.
\end{restatable}

This implies the following, among many similar, corollaries~:
\begin{corollary}
  Unit disk graphs have an $O(n^{1/3 + \eps})$-bit adjacency labeling scheme, for any $\eps > 0$. 
  Disk graphs (for disks with arbitrary radii) have an $O(n^{1/2 + \eps})$-bit adjacency labeling scheme, for any $\eps > 0$. 
  Unit ball graphs (intersection graphs for congruent balls) in $\reals^d$ have an 
  $O(n^{1-2/(d+1) + \eps})$-bit adjacency labeling scheme, for any $\eps > 0$. 
  Ball graphs (for balls with arbitrary radii) in $\reals^d$ have an $O(n^{1-2/(d+2) + \eps})$-bit 
  adjacency labeling scheme, for any $\eps > 0$. 
\end{corollary}

\begin{corollary}
  Segment intersection graphs have an $O(n^{1/3 + \eps})$-bit adjacency labeling scheme, for any $\eps > 0$.
\end{corollary}

The second corollary requires justification, since the number of degrees of freedom, also called the \emph{parametric
dimension} of segments in the plane is $4$.
Nevertheless, the predicate that tests whether two segments $e$ and $e'$ intersect one another is a disjunction of
conjunctions, each consisting of four polynomial inequalities, each involving at most two parameters of $e$ and at
most two parameters of $e'$ (each such inequality determines the side of the line that is spanned by one segment that
contains an endpoint of the other segment). In the terminology of \cite{AAEKS25}, we say that the 
\emph{reduced parametric dimension} of segments, with respect to segment intersection, is only $2$.
It is easy to see, and will follow from the analysis given later, that the bound in the theorem uses
$d=2$ instead of $4$, from which the corollary follows. This observation is general, and calls for
a more careful definition of $d$-dimensional semialgebraic graphs, using the reduced
parametric dimension of the objects (see Agarwal \etal~\cite{AAEKS25} and Sauermann~\cite{MR4205109}).

Recall that a graph is semilinear if it is semialgebraic and the $t$ defining polynomials are linear.
We prove that semilinear families admit much more concise implicit representations.

\begin{restatable}{thm}{semilinear}
  \label{thm:semilinear}
  Semilinear families have an $O(\log n)$-bit adjacency labeling scheme.
\end{restatable}

Note that many well-studied graph families are semilinear: cographs, interval graphs, permutation graphs,
bounded boxicity graphs, circle graphs (intersection graphs of chords of a circle), distance-hereditary graphs
(that preserve graph distances under taking connected induced subgraphs).
While $O(\log n)$-bit adjacency labeling schemes are already known separately for the mentioned classes~\cite{MR1971502,MR2016489,MR4485472}, we obtain this result for all semilinear graphs in an unified way.
In a nutshell, the proof of Theorem~\ref{thm:semilinear} involves a decomposition of the semilinear graph
into a constant number of \emph{comparability graphs} of bounded dimension, as is done in Tomon~\cite{MR4591829}
and Cardinal and Yuditsky~\cite{CY25}, each of which is encoded by a simple scheme involving the rank of
each vertex in each of the dimensions.

We also consider implicit representations of visibility graphs~\cite{dBCvKO08,G07,OR17}.
\emph{Polygon visibility graphs} are graphs whose vertices are in bijection with the vertices of a simple closed polygon,
so that two vertices are adjacent if and only if the line segment between them is contained in the polygon.
Note that polygon visibility graphs do not form a semialgebraic family, in the sense defined above, 
since the adjacency of two vertices depends
on the location of the other vertices. In that case, however, a well-known result of Agarwal, Alon, Aronov, and
Suri~\cite{MR1298916} allows us to obtain an efficient labeling scheme.

\begin{restatable}{thm}{polygon}
  \label{thm:polygon}
  Polygon visibility graphs have an $O(\log^3 n)$-bit adjacency labeling scheme. 
\end{restatable}

We also obtain a slightly better $O(\log^2 n)$ bound for \emph{capped graphs}, a closely related family studied
by Davies, Krawczyk, McCarty, and Walczak~\cite{MR4562782}, which contains \emph{terrain visibility graphs},
where the polygon is replaced by an unbounded $x$-monotone polygonal curve.

\subsection{Discussion} \label{sec:disc}

\begin{figure}
  \begin{center}
    \includegraphics[page=4]{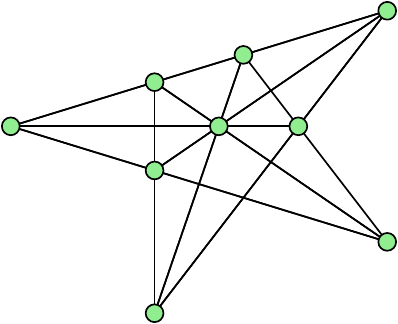}
  \end{center}
  \caption{\label{fig:perles}{\small {\sf The Perles configuration of 9 points and 9 lines, every 
    realization of which requires at least
    one irrational coordinate. This provides an example of a semialgebraic family, the bipartite point-line incidence graphs,
    for which the natural implicit representation by real numbers cannot be turned into an adjacency labeling scheme with
    labels using a bounded number of bits.}}}
\end{figure}

Note that semialgebraic graphs are already defined by an implicit representation involving points in $\reals^d$.
This representation, however, is not easily converted into a binary labeling scheme: 
There exist semialgebraic graphs, all integer realizations of which require coordinates with doubly exponentially
many bits, hence labels of exponentially many bits.
This happens even in the special case of simple geometric intersection graphs such as disk and segment intersection
graphs, see respectively McDiarmid and M\"uller~\cite{MR2995722}, and Kratochv\'il and Matou\v sek~\cite{MR1305055}.
This phenomenon can be traced back to early examples of order types requiring 
coordinates with doubly exponentially many bits, due to Goodman and Pollack~\cite{GPS89}.
In fact, the situation is even worse: There exist semialgebraic families for which the natural implicit representation
by real numbers is not always rational, hence cannot always be (na\"ively)
encoded by integer numbers. This is the case, for instance, of bipartite point-line incidence graphs, 
as witnessed by the famous Perles configuration~\cite{Z08} given in Figure~\ref{fig:perles}.
These examples all relate to the algebraic universality phenomenon first described by Mn\"ev~\cite{MR970093, MR1116375}.

Biclique decompositions of semialgebraic graphs have been studied by Do~\cite{MR4013919}.
She showed that $d$-dimensional semialgebraic graphs have biclique decompositions of size $O(n^{2d/(d+1) + \eps})$.
The same statement, together with construction algorithms, has been given in the context of semialgebraic range
searching by Agarwal et al.~\cite[Appendix A]{AAEKS25}, and in the special case of graphs defined by a single
polynomial by Cardinal and Sharir~\cite{MR4921549}. Note that for a labeling scheme defined as above from a biclique
decomposition, the sum of the lengths of the labels over all vertices is equal, up to a logarithmic factor, to the
size of the biclique decomposition. The best possible length we can expect to achieve with this method is therefore
of the order of $n^{2d/(d+1)+\eps} / n = n^{1-2/(d+1)+\eps}$, for any $\eps>0$,
when every vertex belongs to approximately the same number of
bicliques of the decomposition. Our result shows that this is always possible. In fact, we observe that several 
previously studied methods for constructing biclique decompositions actually produce balanced biclique 
decompositions, in this sense.

It is known that $d$-dimensional semialgebraic families have speed $2^{\Theta (n\log n)}$ for any fixed $d$~\cite{MR4205109},
hence this only provides a logarithmic lower bound on the length of adjacency labels. 
We leave open the problem of finding any better lower bound on adjacency labeling schemes for semialgebraic families,
or alternatively improve on the upper bound stated in Theorem~\ref{thm:main}.

The existence of sublinear-size labeling schemes for semialgebraic families was known before.
Fitch~\cite{MR3944588} was the first to prove the existence of $O(n^{1-\delta})$-bit adjacency labeling schemes for any
semialgebraic family of graphs, for some $\delta$ depending on the family. In particular, he gave an upper 
bound of $O(n^{0.976723})$ for unit disk graphs. Alon~\cite{MR4800729} proved the following~:

\begin{thm}[Alon~\cite{MR4800729}]
  \label{thm:sparseh}
  For every $\eps >0$, there is an integer $d\geq 1$ such that if a hereditary family of graphs has speed at most
  $2^{(1/4-\eps)n^2}$, then the family has an $O(n^{1-1/d}\log n)$-bit adjacency labeling scheme.
\end{thm}

While this result does not directly address adjacency labelings of semialgebraic graphs, the proof actually yields a
similar bound for $d$-dimensional semialgebraic families. 
We state the following result, weaker than our Theorem~\ref{thm:main},
as a consequence of Alon's proof of Theorem~\ref{thm:sparseh}.

\begin{thm}
  \label{thm:weaker}
  $d$-dimensional semialgebraic families have an $O(n^{1-1/d}\log n)$-bit adjacency labeling scheme.
\end{thm}

\begin{proof}
Given a collection of binary vectors of length $n$, the \emph{shatter function} of this family is the function
$g(m)$, for $m\in [n]$, that gives the maximum number of distinct projections of the vectors on $m$ coordinates,
where the maximum is taken over all choices of $m$ coordinates.
Let us now consider the rows of the adjacency matrix of a $d$-dimensional semialgebraic graph.
We claim that these vectors have a polynomially-bounded shatter function.

\begin{lemma}
  \label{lem:shatter}
  The shatter function $g(m)$ of the rows of the adjacency matrix of a $d$-dimensional semialgebraic 
graph satisfies $g(m) = O(m^{d})$.
\end{lemma}
\begin{proof}
  Let $q_1,\ldots,q_m$ be the points that determine the $m$ coordinates (i.e., columns) of the adjacency
matrix. Then, for a point $p$ that determines some row, the signs of the elements of the $p$-row at
columns $q_1,\ldots,q_m$ are determined by the signs of the $mt$ $d$-variate polynomials (in $p$)
$f_i(p,q_j)$, for $i=1,\ldots,t$ and $j=1,\ldots,m$, at the point $p$. The so-called \emph{sign pattern} of 
this sequence of polynomials is upper bounded by the number of cells of their arrangement, which, 
by classical Milnor-Thom type bounds (such as Warren's theorem~\cite{MR226281}), is 
$O((mt)^d) = O(m^d)$, since $t$ is a constant and the polynomials are of constant degree.
\end{proof}

We now use the fact that if a family of vectors has a shatter function $g(m)=O(m^d)$, then they can 
be efficiently encoded in our sense. The following lemma, used by Alon~\cite{MR4800729}, was proved 
by Welzl~\cite{MR1213461} and reformulated in those terms by Alon, Moran, and Yehudayoff~\cite{MR3733359}.
For a binary vector $v$, we refer to a pair of consecutive indices $i$ and $i+1$ such that 
$v_i\not= v_{i+1}$ as a \emph{switch}.

\begin{lemma}
  \label{lem:switches}
  Consider a family of binary vectors of length $n$ with shatter function $g(m) \leq c\cdot m^d$ for
  each $m$, for some constant $c > 0$ and integer $d \geq 1$. Then there is a fixed permutation of the 
  coordinates of the vectors such that for each permuted vector $v$, the number of switches in $v$ 
  is at most $O(n^{1-1/d})$.
\end{lemma}

The adjacency labeling is obtained by first applying such a permutation to the columns of the adjacency matrix, and then
encode each row by the locations of the switches and the sign of the first element. Combining Lemmas~\ref{lem:shatter} and 
\ref{lem:switches} thus yields an upper bound of $O(n^{1-1/d}\log n)$ for $d$-dimensional semialgebraic 
families, namely, the sign of the first element, the sequence of switch locations, and the
column index of the vertex. This completes the proof of Theorem~\ref{thm:weaker}.
\end{proof}

In order to improve the $O(n^{1-1/d}\log n)$ bound of Theorem~\ref{thm:weaker} to $O(n^{1-2/(d+1)+\eps})$, 
as stated in our Theorem~\ref{thm:main}, we use polynomial partitioning.

\subsection{Polynomial partitioning}

Geometric partitioning techniques, allowing multidimensional divide-and-conquer algorithms, have been
fundamental building blocks in the computational geometry literature since the field’s pioneering era.
Seminal contributions by 
Chazelle and Welzl~\cite{MR1014739}, 
Matou\v{s}ek~\cite{MR1174360}, 
and Chazelle~\cite{MR1194032}, 
among others, have established the standard terminology of partition trees and cuttings.
Partition trees were further improved by Chan~\cite{MR2901245}. 
We refer to the survey of Agarwal~\cite{MR3726592} for a thorough account of this line of work through the
lens of range searching.

In their paper on the Erd\H{o}s distinct distances problem, Guth and Katz~\cite{MR3272924} proposed a novel
geometric partitioning scheme that later proved useful for many other problems.
Their main statement is the following~: For any finite set of points in $\reals^d$ and a parameter $D$, there
exists a polynomial $g\in\reals [x_1,\ldots ,x_d]$ of degree $O(D)$, such that each connected component of
$\reals^d\setminus Z(g)$, where $Z(g)$ is the zero set of $g$, contains at most a fraction $1/D^d$ of the
points. In a follow-up paper~\cite{MR3413420}, Guth proved that, given a finite set of low-degree $k$-dimensional
varieties in $\reals^d$ and a parameter $D$, there exists a polynomial of degree $O(D)$
such that each component of $\reals^d\setminus Z(g)$ intersects at most a fraction
$1/D^{d-k}$ of the varieties. These results were later applied successfully in many contexts, including
incidences bounds~\cite{MR3473668,MR3614779,MR2946447}, 
classical theorems in discrete geometry~\cite{MR2957631}, 
unit distances problems~\cite{MR2942731,MR3118901},
and regularity lemmas~\cite{MR4863535}, to name a few. 

Computational aspects of polynomial partitioning, including the delicate recursive handling of situations
in which many points of the input set lie on the zero set $Z(g)$ of the partitioning polynomial $g$ (for
which the results of \cite{MR3413420,MR3272924} provide no useful bound), have
been investigated and resolved by Agarwal, Aronov, Ezra, and Zahl~\cite{MR4243669}, and Matou\v{s}ek and
Pat\'{a}kov\'{a}~\cite{MR3351757}. This allowed the application of polynomial partitioning techniques to
range searching with semialgebraic sets~\cite{MR3123833}, 
3SUM-hard geometric problems~\cite{MR4488010}, 
and algebraic degeneracy testing~\cite{MR4921549}, 
among others.

Our main result heavily relies on polynomial partitioning, in the extended setups of
\cite{MR4243669,MR3351757}, and provides another perspective on the structure
of partition trees produced by the recursive application of the method.

\subsection{Plan of the paper}

Section~\ref{sec:graph} is dedicated to the proof of Theorem~\ref{thm:main}.
The statements in Theorems~\ref{thm:semilinear} and \ref{thm:polygon} are proved in 
sections~\ref{sec:semilinear} and \ref{sec:polygon}, respectively. 

\subsection*{Acknowledgements}

Work by Jean Cardinal was supported by the Fonds de la Recherche Scientifique-FNRS under Grant n°~T003325F. Work by Micha Sharir was partially supported by ISF grant~495/23.

The first author would like to thank Noga Alon for a discussion on his results in~\cite{MR4800729}, and Lena Yuditsky for sharing observations on semilinear families. The authors are also grateful to Pankaj Agarwal, Sariel Har-Peled, and the WG reviewers for insightful comments and references.

\section{Semialgebraic graphs} \label{sec:graph}

This section is dedicated to the proof of the following theorem from the introduction, restated here for convenience.

\semialgebraic*

We first observe that we can restrict our attention to the case where adjacency in the semialgebraic graph is given
by the sign of a single polynomial.
Indeed, since we assume that the number $t$ of polynomials and the size of the boolean predicate $\Phi$ are constant,
we can simply encode the vertex by concatenating the $t$ labels corresponding to the $t$ defining polynomials. Then,
for given labels $\lambda$, $\lambda'$, we use their sublabels corresponding to the $i$th polynomial to compute its
sign, repeat this for each $i$, and then apply the function $\Phi$ on these signs to determine the adjacency of the
two vertices (see for instance Chandoo~\cite{MR4603804} for related reasonings). 

We can also assume, at the cost of a logarithmic factor in the length of labels, that the input is a bipartite semialgebraic graph, in which adjacency is evaluated only between vertices in different parts.
We reduce the case of arbitrary graphs to the bipartite case using a simple divide-and-conquer scheme that covers the
edges of the input graph by bipartite graphs, so that every vertex is contained in $O(\log n)$ such bipartite graphs. 

From now on we will assume that the vertices of the considered semialgebraic graph are points in $\reals^d$.
Hence the setup is as follows: We are given a bipartite graph $G=(P\cup S,E)$ with $P,S\subset\reals^d$ and
$E\subseteq P\times S$, and a polynomial $f\in\reals[x_1,\ldots, x_d, y_1,\ldots, y_d]$ of constant degree
such that $(p,s) \in E$ if and only if $f(p, s) \geq 0$.

We first introduce notations for \emph{ranges} that can be thought of as the duals of the points in $P$ and $S$.
We can map every $s\in S$ to the range
\begin{equation} \label{eq:primal-range}
s^* := \{ x\in\reals^d \mid f(x, s) \geq 0 \},
\end{equation} 
so $(p,s)\in E$ if and only if $p\in s^*$.
Similarly, with each $p\in P$, we can associate a range
\begin{equation} \label{eq:dual-range}
 p^* := \{ y\in\reals^d \mid f(p, y) \geq 0 \},
\end{equation} 
and $(p,s)\in E$ if and only if $s\in p^*$.

We then define the dual sets of ranges $P^* := \{ p^* \mid p\in P\}$ and $S^* := \{ s^* \mid s\in S\}$.
The \emph{incidence graph} of a pair $(P, R)$, where $P$ is a set of points and $R$ a set of semialgebraic sets, is
the bipartite graph on $P\cup R$ where $(p,r)\in P\times R$ is an edge if and only if $p\in r$.
The graph $G$ can therefore be defined as the incidence graph of either the pair $(P, S^*)$ or the 
pair $(S, P^*)$, where both graphs are defined in $\reals^d$.

\subsection{Partition trees}

Given a pair $(P, S^*)$ consisting of a set $P$ of $n$ points and a set $S^*$ of $n$ semialgebraic
sets of constant complexity in $\reals^d$, we find a biclique decomposition of the corresponding 
bipartite incidence graph by constructing a \emph{partition tree} $T$ as follows.

Each node $v$ of $T$ is associated with a pair $(P_v, S^*_v)$, with $P_v\subseteq P$ 
and $S^*_v\subseteq S^*$. The root node corresponds to the input pair $(P, S^*)$.
The child nodes of a node $v$ each contains a subset of $P_v$, that together form a partition of $P_v$.
Let $w$ be such a child node, with an associated subset $P_w\subset P_v$.
Let us consider the set $K_w := \{s^*\in S^*_v \mid P_w\subset s^*\}$ of all ranges in $S^*_v$ 
that contain all points of $P_w$.
The pair $P_w\times K_w$ is a biclique in the incidence graph of $(P_v, S^*_v)$, which we include in our output.
Let us now consider the set $S^*_w := \{s^*\in S^*_v\mid s^*\text{ crosses }P_w\}$, where a range \emph{crosses} a
point set if it intersects it but does not contain it. If the number $|P_w|$ of points in $P_w$ is larger 
than some constant, a partition tree is constructed recursively on the pair $(P_w, S^*_w)$.
Otherwise, the remaining $O(|S_w^*|)$ incidences are encoded explicitly. It is clear that all incidences 
of the incidence graph of $(P, S^*)$ are thus encoded in this way in a one-to-one manner.

The same construction can naturally be defined in the dual setting defined by the pair $(S, P^*)$.
Furthermore, the two constructions can be intertwined: Every partitioning step can take place either 
in the primal or in the dual setting.

For our adjacency labeling purpose, we require that for each vertex of the incidence graph 
(i.e., point or range), the number of bicliques in which this vertex appears is bounded by 
an expression that gives, up to a logarithmic factor, the size of the encoding.
In the above partition tree, observe that a range appears in a node $v$ if it crosses the corresponding
point set $P_v$, but it participates in a biclique with the point set $P_w$ associated with a child node
$w$ of $v$ if it contains $P_w$.
On the other hand, the number of bicliques to which a point belongs is equal to the number of nodes
of the partition tree in which it appears.
In what follows, since the number of child nodes of a node will be bounded by a constant,
the number of bicliques in which a range participates will be proportional to the number of nodes 
in which it appears. We can therefore restrict our attention to the number of
nodes in which each point or range appears.

\subsection{Main construction}

We construct a partition tree based on the efficient polynomial partitioning techniques described by 
Matou\v{s}ek and Pat\'akov\'a~\cite{MR3351757}, and by Agarwal et al.~\cite{MR4243669}. 
Our technique is the same as that of Agarwal et al.~\cite[Appendix A]{AAEKS25}, but the analysis 
is different, and significantly simpler. In particular, we only focus on balancing the number of 
nodes in which points and ranges appear in the primal and dual phases, and do not need the full 
space-query time tradeoff for range searching described in \cite{AAEKS25}.
Another difference is that we restrict the analysis to semialgebraic graphs defined by a 
single polynomial, which allows us to avoid recursing on the number of conjunctions in the 
Boolean function $\Phi$ written in conjunctive normal form, as done in \cite{AAEKS25}.

The first relevant lemma pertains to the existence of a partitioning polynomial for a given 
set of points lying in a $k$-dimensional algebraic variety in $\reals^d$, for any $1\le k\le d$.
We denote by $Z(g)$ the zero set of a polynomial $g\in\reals[x_1,\ldots,x_d]$.

\begin{lemma}[Matou\v{s}ek and Pat\'akov\'a~\cite{MR3351757}]
\label{lem:mp}
Let $V$ be an algebraic variety of dimension $1\le k\le d$ in $\reals^d$ such that all of its irreducible components have
dimension $k$ as well, and the degree of every polynomial defining $V$ is at most $E$. 
Let $P \subset V$ be a set of $n$ points, and let $D\gg E$ be a parameter. 

There exists a polynomial $g\in\reals[x_1,\ldots,x_d]$ of degree at most $E^{d^{O(1)}}D$ such that 
\begin{enumerate}
\item $V\cap Z(g)$ has dimension at most $k-1$.
\item Each connected component of $V\setminus Z(g)$ contains at most $n/D^k$ points of $P$.
\end{enumerate}

Assuming $D, E, d$ are constants, the polynomial $g$, a semialgebraic representation of the connected components of
$V\setminus Z(g)$, the points of $P$ lying in each component, as well as the set of points of $P\cap Z(g)$, 
can be computed in $O(n)$ randomized expected time.
\end{lemma}

The second lemma generalizes the above result to semialgebraic sets.

\begin{lemma}[Agarwal~\etal~\cite{MR4243669}]
\label{lem:aaez}
Let $V$ be an algebraic variety of dimension $1\le k\le d$ in $\reals^d$, defined by polynomials of degree at most $E$. 
Let $P^*$ be a set of $n$ semialgebraic sets in $\reals^d$, each of complexity at most some constant $b$, 
and let $D\gg E$ be a parameter.

There exists a polynomial $g\in\reals[x_1,\ldots,x_d]$ of degree at most $E^{d^{O(1)}} D$ such that
\begin{enumerate}
\item $V\cap Z(g)$ has dimension at most $k-1$.
\item Each connected component of $V\setminus Z(g)$ is crossed by (i.e., intersected by but
  not contained in) at most $n/D$ sets of $P^*$.
\end{enumerate}

Assuming $D$, $E$, $d$, and $b$ are constants, the polynomial $g$, a semialgebraic representation of the connected
components of $V\setminus Z(g)$, and the elements of $P^*$ crossing each component can be computed in $O(n)$
randomized expected time.
\end{lemma}

In what follows, given some partitioning polynomial $g$, we will refer to the connected components of 
$V\setminus Z(g)$ as the \emph{cells} of the partition, and to the set $V\cap Z(g)$ as the \emph{special cell}.
(The standard cells are all connected, but the special cell does not have to be connected.)
The two lemmas above respectively give bounds on the number of points contained in each cell and the 
number of semialgebraic sets crossing each cell (except for the special cell). 
A consequence of the well-known Milnor-Thom type series 
of results~\cite{MR1401711,MR161339,MR200942} is that the number of connected components of 
$V\setminus Z(g)$, and hence the number of cells, is at most $(c D)^k$ for some constant $c$ 
depending on $d$ and $E$. As another consequence of the Milnor-Thom type bounds, we can also 
give a bound on the number of cells crossed by any semialgebraic set. To state these bounds, we
reuse the notations of Lemmas~\ref{lem:mp} and \ref{lem:aaez}. That is,

\begin{lemma}
  \label{lem:mt}
  Let $V$ be an algebraic variety of dimension $1\le k\le d$ in $\reals^d$, defined by polynomials of degree at most $E$,
and $g\in\reals[x_1,\ldots,x_d]$ a polynomial of degree at most $E^{d^{O(1)}} D$, for some parameter $D\gg E$.
  Then the number of connected components of $V\setminus Z(g)$ crossed by a semialgebraic set of complexity at most
$b$ is at most $(cD)^{k-1}$, where the constant $c$ depends on $d, E$, and $b$. 
\end{lemma}

\begin{proof}[Proof of Theorem~\ref{thm:main}]
  We proceed in two phases, corresponding to the primal and dual settings.
In the first phase, we construct a partition tree in the primal setting $(S, P^*)$ using Lemma~\ref{lem:aaez}.
We partition recursively until we reach leaves that contain at most $n^{d/(d+1)}$ ranges from $P^*$.
In the second phase, we construct a partition tree for each such leaf in the dual setting $(P,S^*)$, 
using Lemma~\ref{lem:mp}, continuing all the way up to leaves with a constant number of points from $P$.
A schematic description of the partition tree thus obtained is given in Figure~\ref{fig:schema}.

\begin{figure}
  \begin{center}
    \includegraphics[page=3]{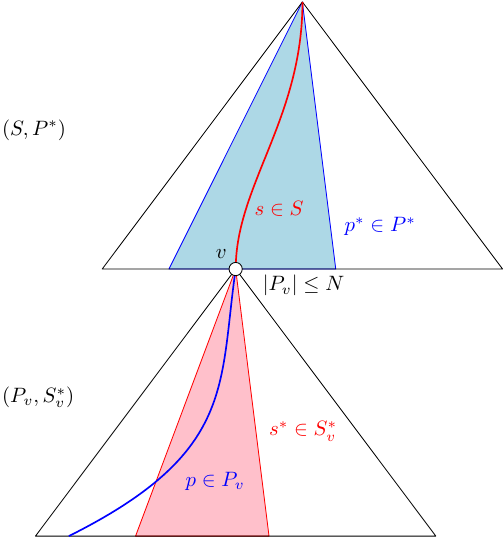}
  \end{center}
  \caption{\label{fig:schema}{\small{\sf A schematic description of the partition tree used in the proof of
Theorem~\ref{thm:main}. The top tree corresponds to the first phase, in which we repeatedly apply
Lemma~\ref{lem:aaez}, and at the end of which we are left with leaves each containing at most 
$N = n^{d/(d+1)}$ elements of $P$. A dual partition tree is constructed from each leaf $v$ 
by repeatedly applying Lemma~\ref{lem:mp}. The blue and red subtrees depict the nodes in which, 
respectively, some element of $P$ (that is, of $P^*$) and some element of $S$ appear.} }}
\end{figure}

In more details, we first consider the pair $(S, P^*)$, where $P^*$ is the set of ranges dual to the points of $P$.
We construct a partition tree by repeatedly applying Lemma~\ref{lem:aaez} with suitably chosen degrees $D$.
Initially, we set $V$ to be the whole space $\reals^d$, and put $k=d$. We recurse on the cells defined 
by the partitioning polynomial $g$, as well as on the set of points contained in the special cell 
$V\cap Z(g)$, and a range is retained in the recursive call only if it crosses the cell. In the case of 
the special cell $V\cap Z(g)$, we do not have any guarantee on the number of points, but 
Lemma~\ref{lem:aaez} ensures that the dimension $k$ of the ambient variety $V$ decreases.
We stop our iterations when the number of semialgebraic sets from $P^*$ remaining in a node is at most $N = n^{d/(d+1)}$.

From Lemma~\ref{lem:mt}, we have that at every cell of this first phase of the partition, 
the number of child subcells intersected by a set
$p^*\in P^*$ is at most $(cD)^{k-1}$, for some constant $c$ (that depends on the much smaller
degrees used earlier in the recursion, for higher-degree varieties), which we assume to be $\ge 1$,
where $D$ is the degree parameter at hand.
We can therefore give an upper bound on the number $\nu (n, k)$ of nodes of the partition tree in which a given range
$p^*\in P^*$ appears, where $n$ is the size of $P^*$, and $k$ is the degree of the ambient variety.
For $k>1$, combining Lemma~\ref{lem:aaez} with Lemma~\ref{lem:mt}, we have
\begin{equation}
  \label{eq:nurec}
  \nu (n, k) \leq (c_kD_k)^{k-1} \nu ( n/D_k, k ) + \nu (n, k-1),
\end{equation}
where the values $D_k$ are the degrees of the partitioning polynomials and each coefficient $c_k$
depends on the earlier degrees, $D_d,\ldots,D_{k+1}$. The degrees $D_k$ are chosen such that $D_{k-1}\gg D_k$.
For this choice to make the subsequent analysis work, we require that $D_k^\eps \ge 2c_k^{k-1}$.
We also have two base cases. First, we have $\nu (N, k) = 1$, for each $k$, since, by definition, 
we stop recursing when at most $N$ sets are left. Then we observe that if $k=1$, we can use balanced 
binary search trees to get $\nu (n, 1) = O(\log n)$. 

The following lemma handles the analysis of the first part of the structure.
We stress that in the lemma, the parameter $N$ is fixed, equal to $n^{d/(d+1)}$ for 
the initial value of $n$, whereas $n$ varies with the recursion.

\begin{lemma} \label{lem:part1}
For each $1\le k\le d$ and each $n\ge N$, we have $\nu (n, k) \leq A_k \left(\frac n N\right)^{k-1+\eps}$, where $\eps = \max_{1\le k\le d} \log_{D_{k}} (2c_k^{k-1})$,
and $A_k$ depends on $\eps$ and $k$.
\end{lemma}
\begin{proof}
Given $\eps$, this dictates the choice of the degrees $D_k$, to ensure that each $D_k$ satisfies
$D_k \ge (2c_k^{k-1})^{1/\eps}$. Once we ensure that this holds for $k=d$, the choice of the other
constants $c_k$ can be enforced recursively, as part of the requirements $D_k\gg D_{k+1}$.
Since the degrees $D_k$ grow rapidly as $k$ decreases to $1$, we need to have $D_1\le n$, 
or else there is no partition for $k=1$. So $n$ has to be at least some (rather large) 
constant value $n_0 = n_0(\eps)$.  When $n \le n_0$ we simply use the trivial bound $\nu(n,k) = n$.

The bound in the lemma clearly holds for the two base cases, if the coefficients $A_k$ are 
chosen sufficiently large. Another base case is $n\le n_0(\eps)$. However, in this part of
the structure we always have $n\ge N$, so this case arises only when $N = n^{d/(d+1)}$
(for the initial value of $n$) is smaller than $n_0(\eps)$. When this is the case, we do 
not recurse at all, and just use the aforementioned bound $\nu(n,k) = n$. 
Since the right-hand side of the bound is at least $A_k$, it suffices
to ensure that $n_0(\eps) \le A_k$, which can be enforced by choosing $A_k$
sufficiently large.

We thus continue the analysis assuming that $n > n_0(\eps)$. We will further assume that 
$A_k \geq 2A_{k-1}$, and proceed by induction on $n$ and $k$, with $n > N$.
The recurrence, combined with the induction hypothesis, gives us: 
  \begin{eqnarray*}
    \nu (n, k) & \leq & (c_kD_k)^{k-1} \nu ( n/D_k, k ) + \nu (n, k-1 ) \\
    & \leq & (c_kD_k)^{k-1} A_k \left(\frac {n} {N D_k}\right)^{k-1+\eps} + A_{k-1} \left( \frac{n}{N}\right)^{k-2+\eps} \\
    & = & A_k \left(\frac n N \right)^{k-1+\eps} \left( \frac{c_k^{k-1}}{D_k^{\eps}} + \frac {A_{k-1}N}{A_k n} \right) \\
    & \leq & A_k \left( \frac n N \right)^{k-1+\eps} \left( \frac 12 + \frac 12 \right) 
    \ \ \text{(since $D_k \ge (2c_k^{k-1})^{1/\eps}$, $n \ge N$, and $A_k\ge 2A_{k-1}$)} \\
    & = & A_k\left(\frac n N\right)^{k-1+\eps} . 
  \end{eqnarray*}
\end{proof}

Letting $k=d$ and taking $n$ to be its original value, this solves to:
\[
\nu (n, d) \leq A_d \left(\frac n N\right)^{d-1+\eps}
= O( n^{1-\frac 2{d+1} + \eps}) .
\]
We also observe that every point in $S$ shows up in exactly one subcell at every recursive level, 
hence the number of nodes of the partition tree containing a given point $s\in S$ is $O(\log n)$.

This describes the first half of our decomposition. We then switch roles and consider the pair $(P, S^*)$, 
where now $S^*$ is a set of ranges, and $P$ is the original set of points.
Each leaf $v$ of the first partition tree corresponds to a set $P_v$ of at most 
$N = n^{d/(d+1)}$ points of $P$, and some (possibly large) set $S^*_v$ of ranges from $S^*$.
That is, the first part of the structure guarantees a uniform bound on $|P_v|$ at each of its leaves $v$,
but there is no control on the size of the corresponding set $S_v$, which could be as large as $n$.
All we know is that $\sum_v |S_v|$, over the leaves $v$, is $n$ (because the sets $S_v$ are pairwise disjoint).

This calls for another approach, where now we construct another partition tree on each leaf $v$
of the first structure, but now the partition is dictated by the point set $P_v$ and not by
the set of ranges, as we did in the first part. Concretely, we repeatedly apply Lemma~\ref{lem:mp} 
to the set $P_v$ at a node $v$, with suitable constants $D_k$, recursing on both the 
connected components of $V\setminus Z(g)$ and on the special cell $V\cap Z(g)$, as before.
We end the process when the leaves contain only $O(1)$ points of $P$.
Again, we consider the number, now denoted as $\nu^* (n,k)$, of nodes of the partition tree 
in which a given range $s^*\in S^*_v$ appears, where $n$ is the number of elements in $P_v$.
Combining Lemmas~\ref{lem:mp} and~\ref{lem:mt}, we obtain the following recursion for $k>1$,
which bears some resemblance to \eqref{eq:nurec}: 
\[
  \nu^* (n, k) \leq (c_kD_k)^{k-1} \nu^* ( n/D^k_k, k ) + \nu^* (n, k-1),
\]
where the values $D_k$ are the degrees of the partitioning polynomials, chosen such that $D_{k-1}\gg D_k$, 
as before, and each value $c_k$ depends, as before, on the preceding degrees $D_j$, for $j > k$. 
We also have the base cases $\nu^* (1, k) = 1$ and $\nu^* (n, 1) = O(\log n)$. 

The analog of Lemma~\ref{lem:part1} in the new context is:
\begin{lemma} \label{lem:part2}
For each $1\le k\le d$ and each $n$, we have $\nu^* (n, k) \leq B_k n^{1-1/k+\eps}$,
where $\eps = \max_{1\le k\le d} \left( \frac{1}{k} \log_{D_{k}} \left( 2^{1/k}c_k^{1-1/k} \right) \right)$, and $B_k$ depends on $\eps$ and $k$.
\end{lemma}
\begin{proof}
We also require, as above, that $n > n_0(\eps)$, with the threshold $n_0(\eps)$ defined similarly,
and use the trivial bound $\nu^*(n,k) = n$ when $n\le n_0(\eps)$.
As before, the bound holds for the base cases, if the $B_k$ are chosen sufficiently large.
For the general case, using induction, we obtain
\begin{eqnarray*}
  \nu^*(n,k) & \leq & (c_kD_k)^{k-1} \nu^* ( n/D^k_k, k ) + \nu^* (n, k-1) \\
             & \le  & (c_kD_k)^{k-1} B_k \left(\frac n {D_k^k}\right)^{1-1/k+\eps} + B_{k-1} n^{1-1/(k-1)+\eps} \\
             & =    & B_k n^{1-1/k+\eps} \left( \frac{ c_k^{k-1} D_k^{k-1}}{D_k^{k-1+k\eps}} + 
\frac{ B_{k-1} n^{1-1/(k-1)+\eps} }{ B_{k} n^{1-1/k+\eps} } \right) \\
& = & B_k n^{1-1/k+\eps} \left( \frac{ c_k^{k-1} }{D_k^{k\eps}} + 
\frac{ B_{k-1} }{ B_{k} n^{1/(k(k-1))} } \right) \\
& \le & B_k n^{1-1/k+\eps} ,
\end{eqnarray*}
where the last inequality holds since
\[
D_k^{k\eps} \ge (2^{1/k}c_k^{1-1/k})^k = 2 c_k^{k-1} ,
\]
and the second term in the parenthesis will be smaller than $1/2$ if $B_k$ is chosen larger than $2B_{k-1}$.
This completes the induction step and establishes the lemma.
\end{proof}

Letting $k=d$ and replacing $n$ by $N=n^{d/(d+1)}$, which is the initial value of $|P_v|$ 
for the second part of the structure, we obtain that every range $s^*\in S^*_v$ appears in at most 
$O( n^{1-\frac 2{d+1} + \frac{d\eps}{d+1}})$ nodes.

We may now account for the contributions of the two phases together.
In the partition tree of the first phase, every point of $S$ appears in $O(\log n)$ nodes 
along a path, and every range of $P^*$ appears in $O(n^{1-\frac 2{d+1} + \eps})$ nodes, by
Lemma~\ref{lem:part1}.
Then at every leaf $v$ of this partition tree, we attach a dual partition tree on the set 
of points $P_v$ and the set of ranges $S^*_v$. In that tree, every point of $P_v$ appears 
in at most $O(\log n)$ nodes along a path, and every range of $S^*_v$ appears in at most 
$O( n^{1-\frac 2{d+1} + \frac{d\eps}{d+1}})$ nodes, by Lemma~\ref{lem:part2}.
(We refer to Figure~\ref{fig:schema} for a schematic description.)
So overall, every element of $S$ appears in $O(\log n + n^{1-\frac 2{d+1} + \eps}) = O(n^{1-\frac 2{d+1} + \eps})$ 
nodes, and every element of $P$ appears in 
$O(n^{1-\frac 2{d+1} + \frac{d\eps}{d+1}} \log n) = O(n^{1-\frac 2{d+1} + \eps})$ 
nodes, with a suitable adjustment of the coefficients, as desired.
This concludes the proof of Theorem~\ref{thm:main}.
\end{proof}

\section{Semilinear graphs}
\label{sec:semilinear}

We now prove our result on semilinear graphs, which we restate:

\semilinear*

Let us first observe that, as briefly mentioned in the introduction,
many well-known families of graphs are semilinear families:
\medskip

\begin{itemize}
\item \emph{Interval graphs} are semilinear, since every interval can be described by two endpoints, 
and testing adjacency involves only comparisons between the endpoints, after an initial sorting of all 
the endpoints. Similarly, \emph{circular arc graphs} are intersection graphs of arcs of a circle, 
and also form a semilinear family.

\item \emph{Permutation graphs} are intersection graphs of line segments whose endpoints lie on two 
parallel lines. Clearly, the intersection pattern can also be described by a semilinear predicate.

\item \emph{Circle graphs} are intersection graphs of chords of a circle (they are a variant of circular 
arc graphs). Again, their intersection can be determined by the relative order of the endpoints around the circle.

\item \emph{Distance-hereditary graphs} are graphs such that graph distances between vertices are preserved 
under taking induced subgraphs. They are known to be circle graphs, hence also form a semilinear family.

\item \emph{Cographs} are graphs that do not have a path on four vertices as an induced subgraph. 
They are known to be permutation graphs, hence are also semilinear.
\end{itemize}
\medskip

We refer to the survey of Brandst\"{a}dt, Le, and Spinrad~\cite{MR1686154} for background on these families.
More generally, intersection graphs of geometric shapes formed by axis-aligned boxes or segments, 
such as \emph{grid intersection graphs}~\cite{MR3817844}, and \emph{max point-tolerance graphs}~\cite{MR3575015}, 
and constant-complexity boolean combinations thereof, form semilinear families.

We first consider two special semilinear families, comparability graphs in dimension $d$, and graphs of boxicity at most $d$, and observe that they have an $O(d\log n)$-bit adjacency labeling scheme.
We then proceed to show that semilinear graphs can be decomposed into constantly many comparability graphs.

\subsection{Comparability graphs and bounded-boxicity graphs}

We define \emph{$d$-dimensional bipartite comparability graphs} as bipartite graphs of the form 
$(S\cup P, E)$, where $S$ and $P$ are set of points in $\reals^d$, and such that $(p,s)\in E$ 
if and only if $p_i<s_i$ for all $i\in [d]$, which we denote by $p\prec s$.
Clearly, $d$-dimensional bipartite comparability graphs are bipartite semilinear graphs.

\begin{lemma}
  \label{lem:comp}
  $d$-dimensional bipartite comparability graphs have an $O(d\log n)$-bit adjacency labeling scheme.
\end{lemma}
\begin{proof}
  Observe that adjacency is determined by the order of the coordinates of the points of $S$ and $P$ 
along each of the $d$ axes. It is therefore sufficient to store, for each point $p$ or $s$ and each 
coordinate axis $i$, the rank of the coordinate $p_i$ or $s_i$ with respect to the other points, 
which requires $O(\log n)$ bits per coordinate. (This requires presorting of the coordinates, 
along each axis separately.)
\end{proof}

Another example of semilinear family is the family of \emph{graphs of boxicity $d$}, which are 
intersection graphs of axis-aligned boxes in $\reals^d$. That is, the vertices of such a graph are
axis-aligned boxes in $\reals^d$, and its edges record pairs of boxes that intersect one another.
Since an axis-aligned box can be specified by two opposite corners, graphs of boxicity $d$ are 
$2d$-dimensional semilinear graphs. Interval graphs are the graphs of boxicity one.
It is known that graphs of boxicity $d$ have an $O(d\log n)$-bit adjacency labeling scheme~\cite[Chapter 4]{MR1971502}.

\subsection{Decomposition}

We reduce the case of arbitrary semilinear graphs to that of bipartite comparability graphs, as is done 
in Tomon~\cite{MR4591829}, and Cardinal and Yuditsky~\cite{CY25}.

\begin{proof}[Proof of Theorem~\ref{thm:semilinear}]
  As shown in~\cite{MR4591829}, we can assume, without loss of generality, that the semilinear 
family is defined by a predicate in the following \emph{disjunctive normal form}:
  \[
  (u,v) \in E \Leftrightarrow \bigvee_{i\in [k]}\left(\bigwedge_{j\in [\ell]} f_{i,j}(u, v) < 0 \right),
  \]
  where the functions $f_{i,j}$ are linear polynomials and $k$ and $\ell$ are constants.
  Note that we can restrict the signs to strict inequalities, a fact which is proved in~\cite{MR4591829}.

  Let $G=(V,E)$ be a graph in this family, and let $G_i=(V, E_i)$, for $i\in [k]$,
  be the graph defined by the $i$th conjunction of inequalities, so 
  $(u,v) \in E_i \Leftrightarrow \bigwedge_{j\in [\ell]} f_{i,j}(u, v) < 0$. Since each $f_{i,j}$ is linear, 
  it can be rewritten as $f_{i,j}(x,y) = g_{i,j}(x) + h_{i,j}(y)$ for some linear functions $g_{i,j}$ and $h_{i,j}$.
  Now each inequality of the form $f_{i,j}(u, v) < 0$ can be rewritten as $g_{i,j}(x) < -h_{i,j}(y)$, and the
  conjunction of the $\ell$ inequalities can be rewritten as 
  $\bigwedge_{j\in [\ell]} \left( g_{i,j}(x) < -h_{i,j}(y) \right)$, which is the same as
  $g_i(u) \prec - h_i(v)$, where $g_i(u) = \left( g_{i,1}(u),\ldots,g_{i,\ell}(u) \right)$
  and $h_i(v) = \left( h_{i,1}(v),\ldots,h_{i,\ell}(v) \right)$.
  This proves that $G_i$ is an $\ell$-dimensional bipartite comparability graph.
  From Lemma~\ref{lem:comp}, $G_i$ has $O(\ell \log n)$-bit adjacency labels.
  Now observe that $G$ is the union of the graphs $G_i$ for $i\in [k]$.
  An adjacency labeling for $G$ can therefore be obtained by concatenating the labels for each of 
  the $G_i$, with an overall size of $O(k \ell \log n) = O(\log n)$.
\end{proof}

\section{Visibility graphs}
\label{sec:polygon}

\begin{figure}
  \begin{center}
    \includegraphics[page=7]{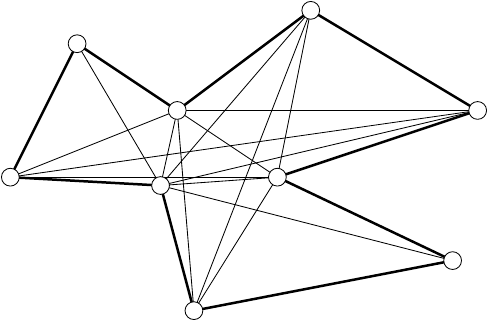}
  \end{center}
  \caption{\label{fig:vis}{\small{\sf The visibility graph of a simple polygon.}}}
\end{figure}

Recall that the visibility graph $G(P)$ of a given simple closed polygon $P$ on $n$ vertices is 
the graph whose set of vertices is the set of vertices of $P$, and such that two vertices are adjacent 
if and only they are \emph{visible} from each other in $P$, i.e. the line segment between them is 
contained in $P$. (For simplicity, although this is not essential, we may assume general position,
so that no segment $uv$ passes through a third vertex of $P$.)
See Figure~\ref{fig:vis}. We refer to Ghosh~\cite{G07}, 
de Berg, Cheong, van Kreveld, and Overmars~\cite{dBCvKO08}, and O'Rourke~\cite{OR17} 
for classical references on visibility graphs.
In this section, we prove the following (restated from the introduction):

\polygon*

\subsection{Decomposition of visibility graphs}

Our proof is by a careful examination of the algorithm described by 
Agarwal, Alon, Aronov, and Suri~\cite{MR1298916} for computing biclique decompositions
of size $O(n\log^3 n)$ of polygon visibility graphs.
It involves a reduction, via point-line duality, to the case of \emph{bichromatic segment intersection graphs}, 
which are bipartite intersection graphs of two sets of (say, red and blue) line segments, such that no two 
segments from the same set intersect in their relative interiors. We summarize this reduction in the following 
statement, whose proof can be found in Chazelle~\cite{MR780415}, and Chazelle and Guibas~\cite{MR1006078}.

\begin{lemma}[Agarwal et al.~\cite{MR1298916}]
  \label{lem:divide}
  Given a simple polygon $P$ on $n$ vertices, there exists a chord of $P$ that partitions $P$ into 
  two subpolygons $P_1$ and $P_2$, each having at most $2n/3$ vertices. Furthermore, the bipartite subgraph 
  of the visibility graph $G(P)$ consisting only of the edges of $G(P)$ having one endpoint in $P_1$ and 
  the other in $P_2$ is, in a suitable dual interpretation, a bichromatic segment intersection graph.
\end{lemma}

The proof of the next statement is inspired by the \emph{hereditary segment tree} construction of Chazelle,
Edelsbrunner, Guibas, and Sharir~\cite{MR1272519}, initially designed for reporting and counting bichromatic 
segment intersections.

\begin{figure}
  \begin{center}
    \includegraphics[page=5, width=\textwidth]{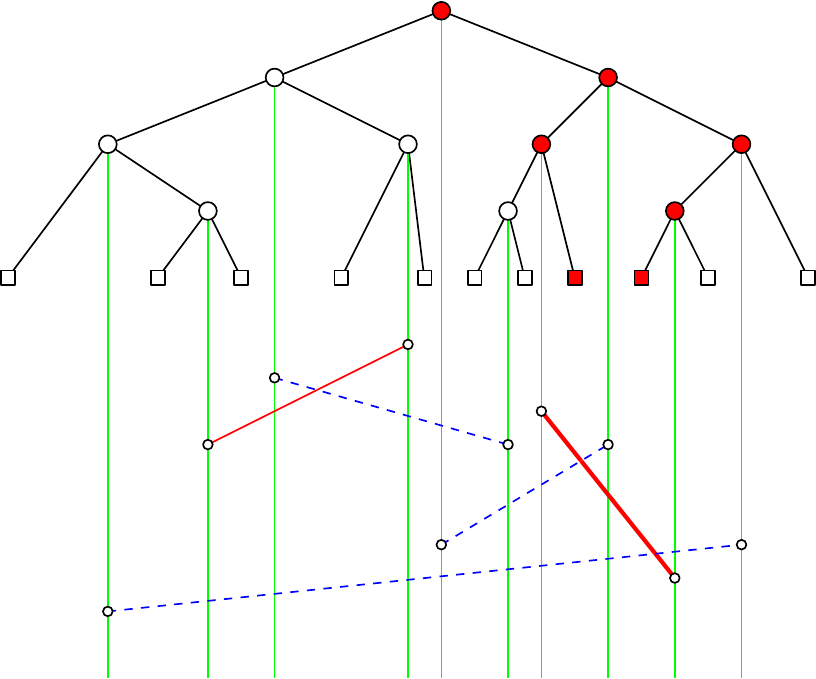}
  \end{center}
  \caption{\label{fig:tree}{\small{\sf A hereditary segment tree for a collection of line segments in the plane.
        The thicker red segment on the right is stored in the filled nodes of the tree, as a long segment in the two leaves,
        and as a short segment in the other nodes.}}}
\end{figure}

\begin{figure}
  \begin{center}
    \includegraphics[page=6]{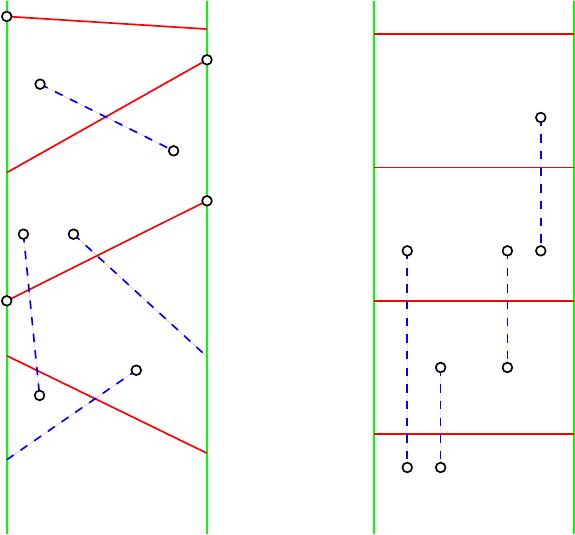}
  \end{center}
  \caption{\label{fig:slab}{\small{\sf A slab containing a number of short blue (dashed) segments, 
    and long red (solid) segments. Their intersection graph only depends on the relative vertical order of the 
    endpoints, with respect to the vertically sorted sequence of the long segments, 
    as shown in the alternative intersection model on the right. The resulting graph is therefore 
    a two-dimensional bipartite comparability graph, with suitable interpretations
    of comparisons and their signs.}}}
\end{figure}

\begin{lemma}
  Bichromatic segment intersection graphs have an $O(\log^2 n)$-bit adjacency labeling scheme.
\end{lemma}
\begin{proof}
  Let $G$ be a bichromatic segment intersection graph on $n$ vertices,
  $S$ be the collection of $n$ red and blue line segments corresponding to the vertices of 
  $G$, and $I$ the collection of $n$ red and blue intervals 
  obtained by projecting the red and blue segments of $S$ on the horizontal axis.

  Consider all $2n$ endpoints of the intervals in $I$, which we can assume to be distinct, without loss 
  of generality, and build a balanced binary search tree whose internal nodes 
  are in one-to-one 
  correspondence with the endpoints. The leaves correspond to spaces (called atomic intervals)
  between pairs of successive endpoints, or to the halfline 
  before the first endpoint or after the last one. Every node of the tree therefore corresponds to 
  a vertical slab in the plane, bounded by two vertical lines (possibly at infinity).
  We refer to Figure~\ref{fig:tree} for an illustration.

  Each segment $s$ of $S$ is stored at the following nodes of the binary tree:
  (i) at the nodes whose slab contains at least one endpoint of $s$ in its interior, as a \emph{short} segment,
  (ii) at the nodes whose slab does not contain any endpoint $s$ in its interior, but whose parent slab does, 
  as a \emph{long} segment. We make the following two observations.
  \begin{observation}
    \label{obs:log}
    Every segment of $S$ is stored at $O(\log n)$ nodes.
  \end{observation}

  \begin{observation}
    \label{obs:inter}
    Every intersection between a red and a blue segment occurs in a slab in which exactly one of the two segment is stored as a long segment.
  \end{observation}

  In order to encode these intersections, we consider each slab, and the intersection graph of
  (i) the blue segments stored as short segments in the corresponding node,
  (ii) the red segments stored as long segments in the corresponding node,
  or symmetrically, with the role of red and blue segments exchanged.
  Our key observation is that these intersection graphs are two-dimensional bipartite
  comparability graphs (in fact, bipartite permutation graphs), see Figure~\ref{fig:slab} for an illustration.
  Hence from Lemma~\ref{lem:comp}, they have an $O(\log n)$ adjacency labeling scheme.

  It remains to combine the labels corresponding to the different slabs together.
  From observation~\ref{obs:inter}, this will allow us to recover all the edges of $G$.
  Every node of the tree can be identified by an $O(\log n)$ bit label.
  We can therefore, for every vertex of $G$, hence segment of $S$, encode the list of nodes of
  the tree in which it appears, and from observation~\ref{obs:log}, this will use at most
  $O(\log^2 n)$ bits. For each of these at most $O(\log n)$ nodes, we provide the $O(\log n)$-bit word encoding adjacency
  in the two-dimensional bipartite comparability graphs corresponding to the node.
  Overall, we spend $O(\log^2 n)$ bits per vertex, as claimed.
\end{proof}

The labels of the bichromatic segment intersection graph, which encode, in dual form, the visibility
segments in $P_1\times P_2$, can be combined with the labels obtained by
recursing on the visibility graphs of the subpolygons $P_1$ and $P_2$ (see Lemma~\ref{lem:divide}). This
results in an additional logarithmic factor, hence $O(\log^3 n)$-bit labels. This finishes the proof 
of Theorem~\ref{thm:polygon}.

\subsection{Decomposition of capped graphs}

A closely related class of graphs is that of \emph{terrain visibility graphs}, which are visibility graphs 
defined on an $x$-monotone polygonal curve (referred to in the literature as a \emph{$1.5$-dimensional terrain}), 
where two vertices are said to be visible from each other if and only if the line segment between them 
lies above the curve~\cite{MR4117719,MR4329054,MR3355568,MR4623911}; again, for simplicity, we
assume here general position.

\begin{figure}
  \begin{center}
    \includegraphics[page=8, width = \textwidth]{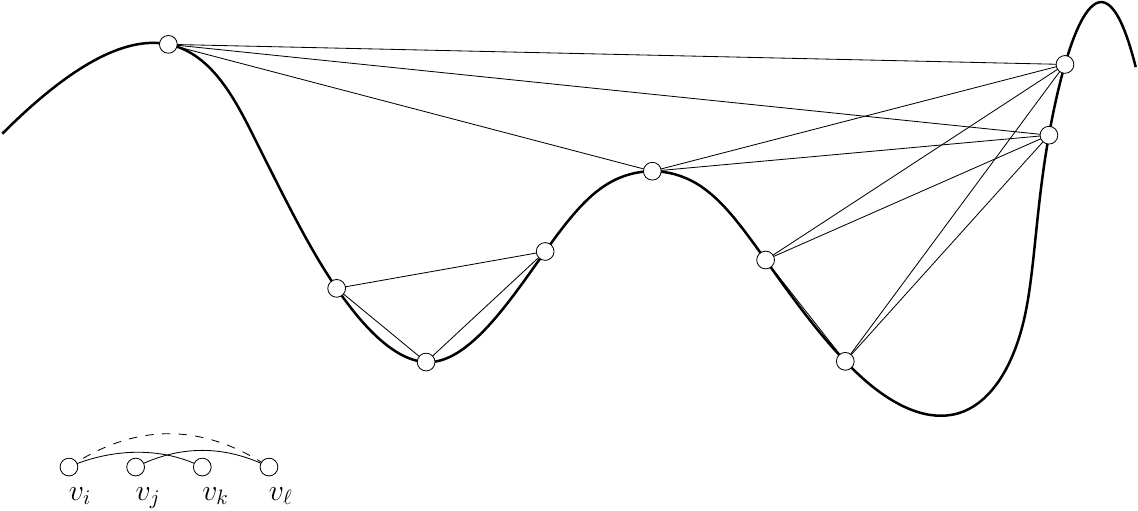}
  \end{center}
  \caption{\label{fig:terrain}\small \sf Example of a visibility graph of points on an $x$-monotone curve. Such visibility graphs are capped graphs, since for any ordered 4-tuple of vertices $v_i$, $v_j$, $v_k$ and $v_{\ell}$, if both pairs of vertices $v_iv_k$ and $v_jv_{\ell}$ are adjacent, then so is $v_iv_{\ell}$, as illustrated in the bottom left of the figure.}
\end{figure}

Recently, Davies, Krawczyk, McCarty, and Walczak~\cite{MR4562782} defined a \emph{capped graph} as 
a graph whose vertex set can be ordered in such a way that for any ordered 4-tuple of vertices
$v_i, v_j, v_k, v_{\ell}$ with $v_i < v_j < v_ k < v_{\ell}$, if both $v_iv_k$ and $v_jv_{\ell}$ are edges, then 
$v_iv_{\ell}$ is also an edge. It is not difficult to observe that terrain visibility graphs, 
with the vertices ordered along the terrain, are capped graphs. 
In fact, it was shown in~\cite{MR4562782} that capped graphs are exactly visibility graphs of points on a 
(not necessarily polygonal) $x$-monotone curve. See Figure~\ref{fig:terrain} for an illustration. 
Observe that capped graphs form a hereditary class, unlike terrain or polygon visibility graphs.
It is known that there exist capped graphs that have no realization as polygonal terrain visibility graphs,
see Ameer, Gibson-Lopez, Krohn, Soderman, and Wang~\cite{MR4117719} for an example.

Cardinal and Yuditsky~\cite{CY25} proved that a bisection algorithm similar to the one described 
in the previous subsection produces near-linear-size biclique decompositions of capped graphs. 
In that case, the bipartite graph processed at every recursive step can be shown to be a two-dimensional 
comparability graph. More precisely, Cardinal and Yuditsky showed the following.

\begin{lemma}
  \label{lem:compfromcapped}
  Let $G$ be a capped graph on $n$ vertices, with vertex ordering $v_1<v_2<\ldots <v_n$.
  For any $\ell\in [n]$, consider the subgraph of $G$ consisting of the edges of the form $v_iv_j$, where $i<\ell$ and $j\ge\ell$.
  Then this subgraph is a two-dimensional comparability graph.
\end{lemma}

From Lemma~\ref{lem:comp}, two-dimensional 
comparability graphs have adjacency labels of size $O(\log n)$.
Hence we can choose $\ell = \lfloor n/2\rfloor$, encode the two-dimensional comparability graph given by
Lemma~\ref{lem:compfromcapped} using Lemma~\ref{lem:comp},
and recurse on the two subgraphs induced respectively by the vertices $v_i$ with $i<\ell$ and $i\ge\ell$.
Multiplying by the logarithmic factor incurred from the recursion, we obtain the following.

\begin{thm}
    Capped graphs have an $O(\log^2 n)$-bit adjacency labeling scheme. 
\end{thm}

\bibliography{Adjacencylabels}

\end{document}